\documentclass[11pt]{amsart}
\usepackage{amsmath}
\usepackage{amsfonts}
\usepackage{amssymb}
\newtheorem{thm}{Theorem}[section]

\newtheorem{cor}[thm]{Corollary}

\newtheorem{prop}[thm]{Proposition}
\newtheorem{prob}[thm]{Problem}
\theoremstyle{definition}
\newtheorem{defn}[thm]{Definition}
\newtheorem{remark}[thm]{Remark}

\def\ket#1{| #1 \rangle}
\def\bra#1{\langle #1 |}

\newcommand{\bb}[1]{\mathbb{#1}}
\newcommand{\cl}[1]{\mathcal{#1}}

\begin{document}

\title[]{Bisynchronous Games and Factorizable Maps}

\author[V.~I.~Paulsen]{Vern I.~Paulsen}
\address{Institute for Quantum Computing and Department of Pure Mathematics, University of Waterloo,
Waterloo, ON, Canada  N2L 3G1}
\email{vpaulsen@uwaterloo.ca}
\author[M.~Rahaman]{Mizanur Rahaman}
\address{Department of Mathematics, BITS Pilani K.K Birla Goa Campus, Goa, India 403726}
\email{mizanurr@goa.bits-pilani.ac.in}


\begin{abstract} We introduce a new class of non-local games, and corresponding densities, which we call {\it bisynchronous}.  Bisynchronous games are a subclass of synchronous games and exhibit many interesting symmetries when the algebra of the game is considered. We develop a close connection between these non-local games and the theory of quantum groups which recently surfaced in studies of graph isomorphism games.  When the number of inputs is equal to the number of outputs, we prove that a bisynchronous density arises from a trace on the quantum permutation group. 

Each bisynchronous density gives rise to a completely positive map and we prove that these maps are factorizable maps. 

\end{abstract}

\maketitle


\section{Introduction}

Recent results have shown that there is a close connection between densities that correspond to perfect strategies for the graph isomorphism game, introduced by \cite{mancinscaetal},  and traces on quantum permutation groups \cite{mancinscaetal2, lupini, brannanetal}.  In this paper, we identify a property of the graph isomorphism game, which we call {\it bisynchronicity} and prove that any game or conditional probability density $p(a,b|x,y)$ that has this property is connected with traces on the quantum permutation group. In addition, we prove that each such density $p(a,b|x,y)$ gives rise to a map that is {\it factorizable} in the sense of Anatharaman-Delaroche\cite{delaroche} as studied by Haagerup and Musat\cite{hm11}.

We begin by briefly recalling the theory and definitions of non-local games and introducing the games that we shall call {\it bisynchronous}.  In section 2, we prove some basic properties about bisynchronous games and bisynchronous conditional probability densities. In section 3, we look at properties of the completely positive map associated with a bisynchronous density. Then in section 4, we turn our attention to the relation between these densities and factorizable maps.

\subsection{Definitions of games and strategies}

By a {\it two-person  finite input-output game} we mean a tuple
 $\cl G=(I_A, I_B, O_A, O_B, \lambda)$ where $I_A, I_B, O_A, O_B$ are finite sets and
\[ \lambda: I_A \times I_B \times O_A \times O_B \to \{ 0,1 \} \]
is a function that represents the rules of the game, sometimes called the {\it predicate}.
The sets $I_A$ and $I_B$ represent the inputs that A and B can receive. These are often thought of as the questions that they can be asked.
The sets $O_A$ and $O_B$, represent the outputs that Alice and Bob can produce,
respectively, often thought of as the possible answers that they can give. When $\lambda(x,y,a,b)=1$, then A and B are considered to have given a satisfactory pair of answers, $a,b$ to the questions $x,y$, and when $\lambda(x,y,a,b)=0$ then the answer is considered to be "wrong" or unsatisfactory.

For each round of the game, the Referee(sometimes also called the {\it verifier}) selects a pair $(x,y) \in I_A \times I_B$, gives Alice $x$ and Bob $y$, and they then produce outputs (answers), $a \in O_A$ and $b \in O_B$, respectively.
They win the round if $\lambda(x,y,a,b) =1$ and lose otherwise.
A and B know the sets and the function $\lambda$
and cooperate before the game to produce a strategy for providing outputs,
but while producing outputs, A and B only know their own inputs
and are not allowed to know the other person's input. Each time that they are given an input and produce an output is referred to as a {\it round} of the game.

These games are memoryless, in the sense that if they receive the same input pair $(x,y)$ on two different rounds, then there is no requirement that they produce the same output pair for both rounds.

Such a game is called {\it synchronous} provided that: (i) A and B have the same input sets and the same output sets, which we denote by $I$ and $O$, respectively, and (ii) $\lambda$ satisfies:
\[ \forall v \in I, \,\, \lambda(v,v,a,b) = \begin{cases} 0 & a \ne b\\ 1 & a=b \end{cases}.\]
That is, whenever Alice and Bob receive the same inputs then they must produce the same outputs. To simplify notation we write a synchronous game as $\cl G= (I,O, \lambda)$.

The concept of a synchronous game was introduced in \cite{PSSTW}, where it was recognized that many games arising from graph theory share this property.

A {\it graph} $G$ is specified by a vertex set $V(G)$ and an edge set $E(G) \subseteq V(G) \times V(G)$, satisfying $(v,v) \notin E(G)$ and $(v,w) \in E(G) \implies (w,v) \in E(G)$. 

Given two graphs $G$ and $H$, a {\it graph homomorphism from G to H} is a function $f:V(G) \to V(H)$ with the property that $(v,w) \in E(G) \implies (f(v), f(w)) \in E(H)$.
The {\it graph homomorphism game} from $G$ to $H$, denoted $Hom(G,H)$ has inputs $I_A=I_B = V(G)$ and outputs $O_A=O_B = V(H).$  They win provided that whenever Alice and Bob receive inputs that are an edge in $G$, then their outputs are an edge in $H$ and that whenever Alice and Bob receive the same vertex in $G$ they produce the same vertex in $H$. For any pair of graphs, this is a synchronous game.

Finally, it is not difficult to see that if $K_c$ denotes the complete graph on $c$ vertices then a graph homomorphism exists from $G$ to $K_c$ if and only if $G$ has a c-coloring. This is because any time $(v,w) \in E(G)$ then a graph homomorphism must send them to distinct vertices in $K_c$. Similarly, a graph homomorphism exists from $K_c$ to $G$ if and only if $G$ contains a $c$-clique, i.e., a subset of $c$ vertices that are all mutually connected. A set of vertices in a graph $G$ is {\it independent} if there are no edges connecting the vertices.  Finding an independent set of size $c$ is the same as finding a graph homomorphism from $K_c$ to the {\it graph complement} $\overline{G}$, i.e., the graph with the same vertices but the complementary edge relation.



\begin{defn} We call a synchronous game $\cl G=(I,O, \lambda)$ {\bf bisynchronous} provided that, whenever $x \ne y$, then $\lambda(x,y,a,a) =0$.
\end{defn}

Thus, a game is bisynchronous provided that A and B have the same question and answer sets and whenever A and B receive the same question they must give the same answer and whenever they receive different questions, then they must give different answers.

It is not hard to see that $Hom(G, K_c)$ need not be bisynchronous, since two distinct vertices could be given the same colour.  The game $Hom(K_c,G)$ is bisynchronous, since whenever a pair of vertices in $K_c$ are not equal they are connected and so the output pair must be connected in $G$, and hence not equal, since we do not allow loops. Thus, the game for determining colouring number is not bisynchronous, while the games for determining clique and independence numbers are bisynchronous.

 The {\it graph isomorphism game} from $G$ to $H$, denoted $Iso(G,H)$ was introduced in \cite{mancinscaetal}. It is a synchronous game whose input and output sets are the same and are both the disjoint union of $V(G)$ and $V(H)$.  We refer to \cite{mancinscaetal} for the list of rules. But one rule, in particular, says that the {\it relation} between any pair of inputs $(v,w)$ and outputs $(a,b)$ must be the same. Here, relation refers to whether the pair belongs to the same graph or not and if they belong to the same graph, then they are either equal or form an edge. Thus, the graph isomorphism game is bisynchronous, because if two input vertices are not equal, then the output vertices must not be equal either.

A {\it random strategy} for a two person finite input-output game can be identified with a conditional probability density $p(a,b|x,y)$, which represents the probability that, given inputs $(x,y) \in I_A \times I_B$, A and B produce outputs $(a,b) \in O_A \times O_B$.  Thus, $p(a,b|v,w) \ge 0$ and for each $(v,w),$
\[ \sum_{a \in O_A, b \in O_B} p(a,b|v,w) =1.\]


Given a game $\cl G$, a random strategy is called {\it perfect} if
\[ \lambda(x,y,a,b)=0 \implies p(a,b|x,y) =0, \, \forall (x,y,a,b) \in I_A \times I_B \times O_A \times O_B.\]
Thus, a perfect strategy has probability 0 of producing a "wrong" answer.

In order for $p(a,b|x,y)$ to be a perfect strategy for a synchronous game it must satisfy $p(a,b|x,x) =0$ whenever $a \ne b$. Such densities were called {\it synchronous} in \cite{PSSTW} where they were proven to arise as traces on certain C*-algebras.  

\begin{defn} We shall call a conditional probability density $p(a,b|x,y), \, x,y \in I, a,b \in O$ {\bf bisynchronous}, provided that
$p(a,b| x,x) =0, \, \forall a \ne b$ and $p(a,a| x,y) = 0, \, \forall x \ne y$.
\end{defn}

Thus, the bisynchronous densities are a subset of the synchronous densities.

These conditional probability densities can belong to several different sets of densities.

Definitions of these various sets of probability densities, including {\it loc, q, qs, qa, qc, vect, nsb} can be found in \cite[Section~6]{OP}, \cite{KPS} or in  \cite{PSSTW}, so we will avoid repeating them here.
We only remark that for $t\in \{ loc, q, qa, qc, vect, nsb \}$ we use $C_t$ to denote the corresponding set of conditional probabilities. It is known that
\begin{equation}\label{Eq-heirarchy}
 C_{loc} \subsetneq C_q \subseteq C_{qs}\subseteq C_{qa} \subseteq C_{qc} \subsetneq C_{vect} \subsetneq C_{nsb}.
\end{equation}

The sets $C_q$, $C_{qs}$, $C_{qa}$ and $C_{qc}$ represent four potentially different mathematical models for the set of all probabilities that can arise as outcomes from entangled quantum experiments. The question of whether or not $C_{qa}=C_{qc}$ for any number of experiments and any number of outputs is known to be equivalent to Connes' embedding conjecture due to results of \cite{Oz13}.

When we want to emphasize that the densities have $n$ inputs and $k$ outputs, we will write $C_t(n,k)$.  We will let $C_t^{bs}(n,k)$ and $C_t^s(n,k)$ denote the subsets of $C_t(n,k)$ consisting of bisynchronous and synchronous densities, respectively.  Thus, we have that
\[ C_t^{bs}(n,k) \subseteq C_t^s(n,k) \subseteq C_t(n,k).\]

We say that $p(a,b|x,y)$ is a {\it perfect t-strategy} for a game provided that it is a perfect strategy that belongs to the corresponding set $C_t$ of probability densities.

\begin{remark}\label{synctobisync} There is a way that every synchronous game can be made into a bisynchronous game.  Given a synchronous game $\cl G= (I, O, \lambda)$ we obtain a bisynchronous game simply by demanding that the players also return the question. Formally, let $\overline{\cl G}= ( I, \overline{O}, \overline{\lambda})$ where
$\overline{O} = I \times O$ and
\[ \overline{\lambda}(x,y, (x',a), (y',b)) = \lambda(x,y,a,b) \delta_{x,x'} \delta_{y,y'},\]
where $\delta_{z,w}$ is the Dirac delta function. Now it is easily seen that this game is bisynchronous. 
\end{remark}

\subsection{The Algebra of a synchronous game}
In \cite{OP, vp, HMPS} the concept of the {\it *-algebra of a synchronous game} $\cl A(\cl G)$ was introduced and studied. Given a game $\cl G= (I, O, \lambda )$, its algebra $\cl A(\cl G)$ is the unital  *-algebra with generators $\{ e_{x,a}: x \in I, a \in O \}$ and relations
\begin{itemize}
\item $e_{x,a} = e_{x,a}^2 = e_{x, a}^*, \, \forall x \in I, a \in O$,
\item $\sum_a e_{x,a} =1, \, \forall x \in I$,
\item $\lambda(x,y,a,b)=0 \implies e_{x,a}e_{y,b}=0.$
\end{itemize}

It was proven that the existence of various types of perfect strategies for the game were equivalent to the existence of various types of representations of this algebra. In addition, we say that $\cl G$ has a {\it perfect C*-strategy} provided that there exists a unital *-homomorphism $\pi: \cl A(\cl G) \to B(\cl H)$ for some Hilbert space $\cl H$ ($\cl H\neq 0$), we call such a map a {\it representation of $\cl A(\cl G)$ on $\cl H$}. We say that $\cl G$ has a {\it perfect $A^*$-strategy} provided that $\cl A(\cl G) \ne (0)$.

In \cite{brannanetal} two synchronous games $\cl G_1, \cl G_2$
were called {\it *-equivalent} if there existed unital *-homomorphisms $\pi: \cl A(\cl G_1) \to \cl A(\cl G_2)$ and $\rho: \cl A(\cl G_2) \to \cl A(\cl G_1)$. It follows from \cite[Proposition~5.3]{brannanetal} that if $\cl G_1$ and $\cl G_2$ are *-equivalent, then $\cl G_1$ has a perfect t-strategy if and only if $\cl G_2$ has a perfect t-strategy, for $t \in \{ loc, q, qa, qc, C^*, A^* \}$.

We now show that  the bisynchronous game $\overline{\cl G}$  obtained from a synchronous game $\cl G$ as in Remark~\ref{synctobisync}  is *-equivalent to $\cl G$. Hence, $\overline{\cl G}$ will have a perfect t-strategy for $t \in \{ loc, q, qa, qc, C^*, A^*  \}$ if and only if $\cl G$ has a perfect t-strategy.

To see this claim note that the algebra of $\overline{\cl G}$ will have generators $e_{x, (x',a)}, \, x,x' \in I, a \in O$ satisfying
\begin{itemize}
\item $e_{x,(x',a)} = e_{x,(x',a)}^* = e_{x,(x',a)}^2, \, \forall x,x' \in I, a \in O$
\item $\sum_{(x',a)} e_{x, (x',a)} = 1, \, \forall x$
\item $\lambda(x,y, (x',a), (y',b)) =0 \implies e_{x, (x',a)}e_{y, (y',b)} =0.$
\end{itemize}
From this it follows that if $x \ne x'$ then for any $y$,
\[ e_{x,(x',a)} = \sum_{(y',b)} e_{x,(x',a)}e_{y, (y',b)} =0.\]
Setting $f_{x,a} = e_{x,(x,a)}$, we see that these elements generate $\cl A(\overline{\cl G})$ and satisfy the same relations as the canonical generators $e_{x,a}$ of $\cl A(\cl G)$.
Thus, the map $f_{x,a} \to e_{x,a}$ defines a *-algebra isomorphism of $\cl A(\overline{\cl G})$ onto $\cl A(\cl G)$.

A similar proof, working with densities, shows that $\cl G$ has a perfect non-signalling strategy if and only if $\overline{\cl G}$ has a perfect non-signalling strategy.

Thus, in a certain sense, we can always reduce the study of synchronous games to bisynchronous games.  However, our strongest results about bisynchronous games requires that the number of inputs be equal to the number of outputs and this property will never be satisfied by the bisynchronous game $\overline{\cl G}$ obtained in this fashion, except in the trivial case of one output games.


\section{Properties of Bisynchronous games and bisynchronous correlations}
Let $\cl G=(I,O, \lambda)$ be bisynchronous, set $|I|=n, |O|=k$ and consider the game algebra $\cl A(\cl G)$ \cite{HMPS}.
  Note that the synchronous condition implies that $e_{x,a}e_{x,b} =0, \, \forall x \in I, \, \forall a \ne b$.
 The bisynchronous condition implies that in addition,
\[ e_{x,a} e_{y,a} =0, \, \forall a \in O, \, \forall x \ne y.\]
Let
\[ p_a = \sum_x e_{x,a},\]
we have that $p_a=p_a^*$ and
\[ p_a^2 = \sum_{x,y} e_{x,a}e_{y,a} = \sum_x e_{x,a} = p_a.\]
Also, 
\[ \sum_a p_a = \sum_x \sum_a e_{x,a} = n I.\]
If we let $q_a = 1 - p_a$ then $q_a = q_a^* = q_a^2$ and
\[ \sum_a q_a = \sum_a ( 1 - p_a) = kI - nI.\]
Thus, if $\cl G$ has a perfect C*-strategy, i.e., the game algebra has a unital *-representation as operators on a Hilbert space, or, equivalently, there exist projections $\{E_{x,a} \}$ on some Hilbert space satisfying these relations, then
\[ n \le k.\]
We actually need to assume that $\cl A(\cl G)$ has a representation on a Hilbert space to make this last conclusion since it was shown in \cite{HMPS} that *-algebras exist with self-adjoint idempotents that sum to negative multiples of the identity.

\begin{remark}
Let $n=k$ and assume that the bisynchronous game has a perfect C*-strategy, i.e., that there is a unital *-homomorphism, $\pi: \cl A(\cl G) \to \cl B$ for some unital C*-algebra $\cl B$.  Then we have $\sum_a \pi(q_a)=0$ and hence $\pi(q_a)=0$ for all $a$. So $\pi(p_a)=1$ for all $a$. Which means $\sum_x \pi(e_{x,a})=1$. So the projections $\{\pi(e_{x,a})\}$ satisfy:
\begin{itemize}
\item $\sum_a \pi(e_{x,a})=1$,
\item $\sum_x \pi(e_{x,a})=1$,
\item  $\pi(e_{x,a}) \pi(e_{y,a})=0$ for $x\neq y$ and also $\pi(e_{x,a}) \pi(e_{x,b})=0$ for $a\neq b$. 
\end{itemize}
Note that the above relations imply that if we set, $P=(\pi(e_{x,a}))_{x,a} \in M_n(\cl B)$, then $P$ is what is often called a {\it magic permutation}, {\it magic unitary}, or {\it quantum permutation}. We prefer this latter terminology. Thus, a quantum permutation $P$ is a matrix of projections with $P^*P= PP^*=I_n$, the identity of $M_n(\cl B)$.

The {\it quantum permutation group} $\cl O(S_n^+)$ is the universal C*-algebra generated by elements $\{ u_{i,j} : 1 \le i,j \le n \}$ with relations $u_{i,j}^2= u_{i,j}^* = u_{i,j}, \, \forall i,j$ and
$\sum_j u_{i,j} =1, \forall i$, $\sum_i u_{i,j} =1, \forall j$. For more details on $\cl O(S_n^+)$, see \cite{brannanetal}.

Thus, when $n=k$ and we identify $I=O = \{ 1,..., n \}$, then the image of the algebra of the game inside a C*-algebra is a quotient of the quantum permutation group. In order to avoid this latter identification, it will often be useful for us to consider $\cl O(S_n^+)$ as having generators $\{ u_{x,a}: x \in I, a \in O \}$ where $I$ and $O$ both have cardinality $n$. 
\end{remark}

Recall that if a conditional probability density $p(a,b|x,y)$  is a perfect density for a bisynchronous game, then  $p(a,b|x,y)$ is a bisynchronous density.

We now characterize the bisynchronous densities in the case $n=k$.

\begin{thm}\label{thm-bs-charact.}
 We have the following.
 \begin{enumerate}
\item  $p \in C_{qc}^{bs}(n,n)$ if and only if there is a tracial state $\tau$ on $\cl O(S_n^+)$ such that \[p(a,b|x,y) = \tau (u_{x,a}u_{y,b}).\]
\item $p \in C_{qa}^{bs}(n,n)$ if and only if there is a unital *-homomorphism $\pi$ of $\cl O(S_n^+)$ into an ultrapower of the hyperfinite $II_1$-factor  $\cl R^{\omega}$ such that
\[ p(a,b|x,y) = \tau_{\omega}( \pi(u_{x,a}u_{y,b})),\]
where $\tau_{\omega}$ is the canonical trace on $\cl R^{\omega}.$
\item $p \in C_q^{bs}(n,n)$ if and only if there is a unital *-homomorphism $\pi$ of $\cl O(S_n^+)$ into a finite dimensional C*-algebra with a trace $\tau$ such that
\[ p(a,b|x,y) = \tau( \pi(u_{x,a}u_{y,b})).\]
\item $p \in C_{loc}^{bs}(n,n)$ if and only if there is a unital *-homomorphism of $\cl O(S_n^+)$ into an abelian C*-algebra and a state $\tau$ on the algebra such that
\[ p(a,b|x,y) = \tau(\pi(u_{x,a}u_{y,b})).\]
\item  $p \in C_{vect}^{bs}(n,n)$ if and only if there is a Hilbert space $\cl H$ and vectors $\{ h_{x,a} \}$ such that $h_{x,a} \perp h_{x,b}, \, \forall x, \, \forall a \ne b$, $h_{x,a} \perp h_{y,a}, \forall a, \forall x \ne y$ such that $\forall x,a, \,\, h:=\sum_b h_{x,b} = \sum_y h_{y,a}$ and $h$ is a unit vector, such that
\[ p(a,b|x,y) = \langle h_{x,a}, h_{y,b} \rangle.\]
\end{enumerate}
\end{thm}
\begin{proof}  We first prove the qc case.   Let $\bb F(n,n)$ denote the free product of $n$ copies of the cyclic group of order $n$. This is generated by $n$ unitaries $u_x$ with $u_x^n =1$. Each unitary is of the form $u_x = \sum_a  \omega^a e_{x,a}$ where $\omega = e^{2 \pi i/n}$ and the $e_{x,a}$'s are the spectral projections. Since $p$ is synchronous we know by \cite{PSSTW} that there is a trace $\tau$ on  $C^*(\bb F(n,n))$ such that  $p(a,b|x,y) = \tau(e_{x,a}e_{y,b})$. In particular, $\sum_a e_{x,a} = 1, \, \forall x$. The bisynchronous condition implies that $0 = p(a,a|x,y) = \tau(e_{x,a}e_{y,a})= \tau((e_{y,a}e_{x,a})(e_{x,a} e_{y,a}))$.  Let $\pi$ be the GNS representation of $\tau$ and set $u_{x,a} = \pi(e_{x,a})$. Then each $u_{x,a}$ is a projection and $\sum_a u_{x,a} =1$. On the image of the GNS representation $\pi_{\tau}$ of $\tau$ we will have that $\tau$ is a faithful trace. Thus, from the last equation we see that $\tau( (u_{x,a}u_{y,a})^*(u_{x,a} u_{y,a})) =0$ and it follows that $u_{x,a} u_{y,a} =0$.  By the same argument as in Remark~2.1, we have that $\sum_x u_{x,a} =1$, so that  $P=( u_{x,a})$ is a magic permutation and so the elements $u_{x,a}$ are the image of a representation of $\cl O(S_n^+)$ and $\tau$ induces a trace on this C*-algebra. Thus, every $p(a,b|x,y) \in C_{qc}^{bs}(n,n)$ has the desired form.   

Conversely, it is easy to see that any density of the form $p(a,b|x,y) = \tau( u_{x,a}u_{y,b})$ for some trace on $\cl O(S_n^+)$ is bisynchronous. Thus, (1) follows.

The proofs of (2), (3), and (4) are similar using the characterizations of synchronous correlations in \cite{KPS}.  For example, \cite{KPS} shows that every correlation in $C_{qa}^s(n,n)$ comes from an embedding of $C^*(\bb F(n,n))$ into $\cl R^{\omega}$ composed with the canonical faithful trace $\tau_{\omega}$ on $\cl R^{\omega}$.  If one lets $u_{x,a} \in \cl R^{\omega}$ denote the image of $e_{x,a}$, then the same calculation shows that $P=(u_{x,a})$ is a magic permutation in $M_n(\cl R^{\omega})$. 

We now argue the vectorial case.    By the characterisation of $C^s_{vect}(n,n)$ given in \cite{NGHA} and \cite{NPA} (see also \cite{OP}), we have vectors,  $\{ h_{x,a} \}$ such that $\sum_a h_{x,a} = h$ where $h$ is a unit vector and $h_{x,a} \perp h_{x,b}, \, \forall x, a \ne b$ with $p(a,b| x,y) = \langle h_{x,a} , h_{y,b} \rangle$ where $p$ is our given bisynchronous density.   Set $k_a = \sum_x h_{x,a}$. Since the density is bisynchronous, $h_{x,a} \perp h_{y,a}$ so that 
\[ \|k_a \|^2 = \sum_x \| h_{x,a} \|^2,\]
and
\[ \sum_a \|k_a\|^2  = \sum_{x,a} \| h_{x,a}\|^2 = \sum_x (\sum_a \| h_{x,a}\|^2 ) = \sum_x \| h \|^2 = n.\]
But we also have that
\[ n h = \sum_a k_a \implies n = \sum_a \langle h | k_a \rangle \le \sum_a  \| k_a \| \le (\sum_a \| k_a \|^2 )^{1/2} ( \sum_a 1 )^{1/2} = n^{1/2} n^{1/2}.\]
Thus, we must have equality throughout and this implies that
\[ n= \sum_a \|k_a\| \le \big( \|k_1\|, ..., \|k_n \|\big) \cdot \big( 1,..., 1 \big) \le n.\]
By Cauchy-Schwarz these last two vectors must be parallel which implies that $\|k_a\|= \|k_b \|= 1, \forall a,b$.
Also, since
\[ n = \| \sum_a k_a \| \le \sum_a \|k_a \| = n,\]
it must be the case that $k_a = k_b, \forall a,b$. Finally, since $n h = \sum_a k_a$ we have that $k_a = h, \forall a$.
Thus,  $\sum_y h_{y,a} = h$ and we see that this set of vectors satisfies the conditions of (5).
\end{proof}

\begin{remark} Note that the above proof shows that if we have a set of vectors $\{ h_{x,a} \}$ defining a vectorial bisynchronous density as above then the matrix of vectors $(h_{x,a})_{x \in I, a \in O}$ is a "vector permutation" matrix, in the sense that every row consists of orthogonal vectors, every column consists of orthogonal vectors, and every row and column sums to a fixed unit vector.
\end{remark}



\begin{remark} Recall that the set $C_{qs}(n,k)$ is defined to be the set of correlations arising from tensor products of possibly infinite dimensional Hilbert spaces. In \cite{KPS} it was shown that $C^s_{qs}(n,k) = C^s_q(n,k)$.  Hence, it follows that $C^{bs}_{qs}(n,k) = C^{bs}_q(n,k)$ also.
\end{remark}
 
One can understand bisynchronous games and densities in terms of the {\it flip} operation.

\subsection{The Flip of a Game}\label{subsection-flip}
Given a game $\cl G=(I, O, \lambda)$, we define a new game $ \tilde{\mathcal{G}}=(\tilde{I},\tilde{O}, 
\tilde{\lambda})$, 
where $\tilde{I}=O, \tilde{O}=I$ and \[\tilde{\lambda}(a,b,x,y)=\lambda(x,y,a,b).\] 
It is easy to see that $\tilde{\tilde{\mathcal{G}}}=\cl G$.
Moreover we have the following fact.

\begin{prop} Let $\cl G=( I, O, \lambda)$ be a game.  Then $\cl G$ is bisynchronous if and only if $\cl G$ and its flip $\tilde{\cl G}$ are both synchronous.
Moreover, $\cl G$ is bisynchronous if and only if $\tilde{\cl G}$ is bisynchronous.
\end{prop}

In general the flip of a graph homomorphism game $Hom(G,H)$ is not the game $Hom(H,G)$.  In fact one finds that for the flipped game,  
\[\{ (h,h',g,g'): \tilde{\lambda}(h,h',g,g')=0 \} \supseteq \{ (h, h', g, g'): \mu(g,g',h,h')=0 \} \] 
 and generally strictly larger, where the second set is the tuples that are 0 for $Hom(H,G) := ( V(H), V(G), \mu)$.  In this sense the flipped game has more "rules" that must be obeyed. The following proposition reflects this fact. We mention here the definition of non-signaling densities. 
 
 Recall that we say a game $\cl G$ admits a perfect non-signaling strategy if the corresponding probability density $\{p(a,b|x,y)\}_{a.b.x.y}$ satisfy the following properties:
 \begin{enumerate}
 \item $p(a,b|x,y)\geq 0$ for all $x,y,a,b$ and $\sum_{a,b}p(a,b|x,y)=1$.
 \\
 \item $\sum_{b}p(a,b|x,y)=\sum_{b}p(a,b|x,y')$ for all $y,y'$.
 \\
 \item $\sum_{a}p(a,b|x,y)=\sum_{a}p(a,b|x',y)$ for all $x,x'$.
 \end{enumerate}
We denote the set of non-signaling densities as $C_{nsb}$ and recall from the Equation \ref{Eq-heirarchy} that every class of densities in our context is a subset of $C_{nsb}$.
\begin{prop}
Let $\cl G=Hom(G,H)$, where $G$ is a complete graph and  $H$ is not a complete graph. Then $\tilde{\mathcal{G}}$ has no perfect non-signaling strategy and $\cl A(\tilde{\cl G}) = (0).$
\end{prop}
\begin{proof}
Since $H$ is not a complete graph, there exist $x,y\in V(H)$ such that $x \ne y$ and $(x,y)$ is not an edge. For such a pair $x,y$ we see for any $a,b\in V(G)$
\[\tilde{\lambda}(x,y,a,b)=\lambda(a,b,x,y)=0.\]
The above follows from the fact that $G$ is a complete graph.
If  $p(a,b |x,y)$ is any conditional probability density that satisfies $\tilde{\lambda}(x,y,a,b) = 0 \implies p(a,b |x,y)=0$ then for this particular pair $x,y$ we would have that
$p(a,b |x,y)=0, \, \forall a,b$, contradicting the fact that $1 = \sum_{a,b} p(a,b | x,y)$.

To see the second claim, note that for this pair $x,y$ we have that the corresponding generators  satisfy $e_{x,a}e_{y,b}=0$ for any $a,b$. Hence, 
\[ 1= \big( \sum_a e_{x,a} \big) \big( \sum_b e_{y,b} \big) = \sum_{a,b} e_{x,a}e_{y,b} =0,\]
i.e. the identity must belong to the ideal generated by the relations and hence the algebra is $(0)$.
\end{proof}
\begin{thm} Let $t \in \{ loc, q, qa, qc, vect \}$ and let $p\in C_t^{bs}(n,n)$, then $q(x,y|a,b) := p(a,b | x,y) \in C_t^{bs}(n,n)$.
\end{thm}
\begin{proof} The result follows by invoking the  characterisations of elements of $C_t^{bs}(n,n)$ given in Theorem \ref{thm-bs-charact.}.
\end{proof}

\begin{remark}\label{nscounterex} The above result does not hold for {\it nonsignalling} densities, except in the case $n=2$. Let $p(a,b|x,y)$ for $x,y,a,b \in \bb Z_3$ be given by
\[ p(a,b|x,x) = \begin{cases} 1/3, & a=b \\ 0, & a \ne b \end{cases} \text{ and for } x \ne y, \, p(a,b| x,y) = \begin{cases} 1/3, & a-b = 1 \\ 0, & a-b \ne 1 \end{cases}. \]
Then it is easily checked that $p \in C_{ns}^{bs}(3,3)$. However,  $q(x,y|a,b) := p(a,b|x,y)$ has the property that $q(x,y| 2,0) =0, \, \forall x,y$ and so is not even a probability density.
 \end{remark}




\section{CP maps and bisynchronous correlations}
In this section we look at a certain class of completely positive linear maps which arise from  bisynchronous densities. We follow the idea introduced by Ortiz-Paulsen (see \cite{OP}) for the graph homomorphism game. This idea was further exploited in \cite{mancinscaetal2}.

Let $t \in \{ loc, q, qa, qc, vect, ns \}$.
Given any correlation $p(a,b| x,y) \in C_t(n,k)$, where $x,y \in [n]$ and $a,b \in [k]$ we have a map
\begin{equation}\label{eq-fact-map}
 \Phi_p: M_n \to M_k,   \Phi_p(E_{x,y}) = \sum_{a,b} p(a,b | x,y) E_{a,b}.
\end{equation}
Here $E_{x,y}= \ket{x} \bra{y}$ and $E_{a,b}= \ket{a} \bra{b}$ denote the matrix unit basis for the domain and range, respectively.
We shall call such a map a {\bf t-map}. Note that the density uniquely determines the map $\Phi_p$, i.e., $\Phi_p=  \Phi_q \iff p(a,b|x,y) = q(a,b|x,y), \forall x,y,a,b$. 
\begin{thm}\rm{[Ortiz-Paulsen, see \cite{OP}]}
If $p\in C_{vect}^s(n,k)$ is a synchronous correlation, then the map $\Phi_p$ defined in the Equation (\ref{eq-fact-map}) is completely positive.  
\end{thm}

The above assertion about $\Phi_p$ follows from the fact that the Choi matrix of $\Phi_p$ is a positive semidefinite matrix in $M_n\otimes M_k$. Indeed, since $p\in C_{vect}^s(n,k)$, there exists a Hilbert space $\mathcal{H}$ with vectors $\{h_{x,a}\}$ satisfying the relations $h_{x,a}\perp h_{x,b}, \forall x, \forall a\neq b$ and $\sum_{a}h_{x,a}=\sum_{b}h_{y,b}$, for all $x,y$, such that
\[p(a,b|x,y)=\langle h_{x,a},h_{y,b}\rangle.\] 
Now let $C_{p}:=\sum_{x,y}E_{x,y}\otimes \Phi_p(E_{x,y})$ denote the Choi matrix of $\Phi_p$. If $\{e_{x} \}$ and $\{f_{a}\}$ denote the canonical orthonormal basis for $\mathbb{C}^n$ and $\mathbb{C}^k$ respectively, then for any vector $\xi=\sum_{x,a}\lambda_{x,a}e_{x}\otimes f_{a}\in \mathbb{C}^n\otimes \mathbb{C}^k$ using the above orthogonality relations it follows (see \cite{OP}) that 
\[\langle C_p \xi,\xi\rangle\geq 0.\]

We remark that there are $p \in C^s_{ns}(n,k)$ for which $\Phi_p$ is not completely positive. For example, setting
\begin{multline*} p(0,0|x,y) = p(1,1|x,y) = 1/2, p(0,1 |x,y) = p(1,0 | x,y) =0 \, \forall (x,y) \ne (1,1), \\ p(0,0 | 1,1) = p(1,1|1,1) =0, p(0,1|1,1) = p(1,0 | 1,1) = 1/2,\end{multline*}
we have that $p \in C^s_{ns}(2,2)$ and $\Phi_p$ is not completely positive. 


When $p(a,b|x,y)$ is bisynchronous (respectively, synchronous), we call such a map a {\bf bisynchronous t-map} 
(respectively, {\bf synchronous t-map}). It turns out that for bisynchronous densities, the corresponding map has many interesting properties.

\begin{prop} Let $p \in C_{vect}(n,k)$ be a bisynchronous density, then the map $\Phi_p$
 is a channel (that is, trace preserving and completely positive). If, in addition, $n=k$, then $\Phi_p$ is a unital channel and, if $J$ is the all 1 matrix, then $\Phi_p(J)=J$. Moreover, if $\sigma:M_n\rightarrow \mathbb{C}$ is the functional defined by $\sigma((a_{ij}))=\sum_{i,j}a_{i,j}$, then $\sigma(\Phi_p(A))=\sigma(A)$, for all $A\in M_n$.
\end{prop}

\begin{proof}
The complete positivity follows from exactly similar arguments outlined in \cite{OP}. For trace preserving property, let $x=y$. then
\[Tr(\Phi_p(E_{x,x}))=Tr(\sum_{a,b} p(a,b|x,x)E_{a,b})=\sum_a p(a,a|x,x)=1=Tr(E_{x,x}).\]
Now if $x\neq y$, then 
\begin{align*}
Tr(\Phi_p(E_{x,y}))&=Tr(\sum_{a,b} p(a,b|x,y)E_{a,b})\\
&=Tr(\sum_{a,b(a\neq b)} p(a,b|x,y)E_{a,b})\\
&=0\\
&=Tr(E_{x,y}).
\end{align*} 
Now for $n=k$, we get 
\begin{align*}
\Phi_p(1)&=\Phi_p(\sum_x E_{x,x})\\
&=\sum_{x,a,b} p(a,b|x,x)E_{a,b}\\
&=\sum_{x,a}p(a,a|x,x)E_{a,a}=(\sum_a E_{a,a}\sum_{x}p(a,a|x,x))=\sum_a E_{a,a}=1 
\end{align*}
Here we used $\sum_{x}p(a,a|x,x))=1$ which follows from the fact that $\sum_x e_{x,a}=1$.
Hence the map is a unital and trace preserving map (unital channel).

Now for the last assertion, write $J=\sum_{x,y}E_{x,y}$. Then we calculate
\[\Phi_p(J)=\sum_{x,y}\sum_{a,b}p(a,b|x,y)E_{x,y}=(\sum_{a,b}E_{a,b})(\sum_{x,y}p(a,b|x,y))=\sum_{a,b}E_{a,b}=J.\]
The last assertion can be verified for the matrix units and by linearity it is enough to get the result. Indeed, 
\[\sigma(\Phi_p(E_{x,y}))=\sigma(\sum_{a,b}p(a,b|x,y)E_{a,b})=\sum_{a,b}p(a,b|x,y)=1=\sigma(E_{x,y}).\]
\end{proof}

\begin{remark} The map $\Phi_p$ for the bisynchronous density $p$ given in Remark~\ref{nscounterex}
is easily seen to not even be positive. In fact, if we let $J$ denote the $3 \times 3$ matrix of all 1's, then
\[ \Phi_p(J) = I_3 +2S_3,\]
where $S_3$ denotes the map which cyclically permutes the canonical basis, i.e., $S_3(e_j) = e_{j+1}, \, j \in \bb Z_3$.
\end{remark}

We have the following "composition" rule as in \cite[Proposition~3.5]{OP}.

\begin{prop}
Let $t\in \{loc, q, qa, qc,vect\}$. If $p(x,y|v,w)\in C_t(n,k)$ and $q(a,b|x,y)\in C_t(k,l)$ then
\[r(a,b|v,w):=\sum_{x,y}q(a,b|x,y)p(x,y|v,w)\in C_t(n,l),\]
and $\Phi_r = \Phi_q \circ \Phi_p$. 
Moreover, if $p$ and $q$ are synchronous or bisynchronous, then so is $r$. 
\end{prop} 
\begin{proof} 
Here we only show the bisynchronous property. The other assertions are contained in \cite[Proposition~3.5]{OP}. 
Thus, all that we need to show is if $v\neq w$, then $r(a,a|v,w)=0$.
We compute
\[r(a,a|v,w)=\sum_{x,y}q(a,a|x,y)p(x,y|v,w)=\sum_{x,y(x\neq y)}q(a,a|x,y)p(x,y|v,w)=0.\]
The first equality follow from the fact that $p$ is bisynchronous. The second one follows because $q$ is bisynchronous.

Finally, the composition rule is easily checked.
\end{proof}

\section{Factorizability of bisynchronous maps}

Following Anatharaman-Delaroche \cite{delaroche} and Haagerup-Musat (see \cite{hm11},\cite{hm15}), a map $\Phi:M_n\rightarrow M_n$ is called {\it factorizable} if there exists a von Neumann algebra $N$ with a trace $\tau_N$ and a unitary $u\in M_n\otimes N$ such that
\[\Phi(X)=id\otimes {\tau_N}(u^*(X\otimes 1_N)u),\forall X\in M_n.\]
the algebra $N$ is often called the {\bf ancilla}. Because
we would like to distinguish between different types of ancillas, we adopt the following notations:
\begin{itemize}
\item when $N$ is abelian, we call $\Phi$ {\bf loc-factorizable},
\item when $N$ is finite dimensional, we call $\Phi$ {\bf q-factorizable},
\item when $(N, \tau)$ has a trace preserving embedding into an ultrapower of the hyperfinite $II_1$-factor, we call $\Phi$ {\bf qa-factorizable},
\item we call $\Phi$ {\bf qc-factorizable} when it is factorizable.
\end{itemize}
Note that in \cite{MR18}, the set of factorizable maps on $M_n$ that require a finite dimensional ancilla is denoted by $\mathcal{F}\mathcal{M}_{fin}(n)$. Thus, a map $\Phi: M_n \to M_n$  is q-factorizable in our language if and only if $\Phi \in \cl F \cl M_{fin}(n).$

When a map is t-factorizable and the unitary $u$ can also be chosen to be a quantum permutation, then we call $u$ a {\bf quantum t-permutation} and will say that the map is {\bf t-factorizable via a quantum permutation.}

\begin{thm}\label{prop:factorizable}
Let $t \in \{ loc, q, qa, qc \}$, let $p\in C_{t}(n,n)$  be a bisynchronous density and let $\Phi_p: M_n \to M_n$ be the associated map. Then $\Phi_p$ is  t-factorizable via a quantum permutation.
Conversely, if $\Phi: M_n \to M_n$ is  t-factorizable via a quantum permutation and we write \[\Phi(E_{x,y}) = \sum_{a,b} p(a,b| x,y) E_{a,b},\] 
then $p \in C_t^{bs}(n,n)$ and $\Phi = \Phi_p$.
\end{thm}
\begin{proof}
We only prove the ``qc" case here. The other cases are similar.   Assume that $p \in C_{qc}^{bs}(n,n)$. It is clear from the previous section that the map $\Phi_p$ is a unital channel. Now since $p$ is bisynchronous, from the Theorem \ref{thm-bs-charact.} there is a set of projections $\{e_{x,a}\}$ in a C$^*$-algebra $\mathcal{A}$ with a trace $\tau$ such that $u=(e_{x,a})_{x,a}$ is a quantum permutation in $M_n\otimes \mathcal{A}$ with 
\[p(a,b|x,y)=\tau(e_{x,a}e_{y,b}).\]
 Embedding $(\mathcal{A},\tau)$ into its double dual $\mathcal{A}^{**}$ if necessary, we get a von Neumann algebra $N$ with a trace $\tau_N$ and a set of projections in $N$ which we still call $\{e_{x,a}\}$ and moreover, these projections form a unitary $u=(e_{x,a})$ which is a quantum permutation in $M_n\otimes N$. 
Now we show that \[\Phi_p(X)=id\otimes {\tau_N}(u^*(X\otimes 1_N)u),\forall X\in M_n.\] To this end, we evaluate the map $X\mapsto id\otimes {\tau_N}(u^*(X\otimes 1_N)u)$ on matrix units. Let's relabel the indices and write $u=\sum_{i,j=1}^n  E_{i,j}\otimes e_{i,j}$ and we calculate for any $E_{x,y}$
\begin{align*}
id\otimes {\tau_N}(u^*(E_{x,y}\otimes 1_N)u)&=id\otimes{\tau_N}(\sum_{ij} E_{ij}\otimes e_{j,i}(E_{x,y}\otimes 1_N)\sum_{k,l}E_{k,l}\otimes e_{k,l})\\
&=id\otimes \tau_N(\sum_{i,k,l,j=x}E_{i,y}E_{k,l}\otimes e_{x,i}e_{k,l})\\
&= id\otimes \tau_N(\sum_{i,l}E_{i,l}\otimes e_{x,i}e_{y,l})\\
&=\sum_{i,l}\tau_N(e_{x,i}e_{y,l})E_{i,l}\\
&=\sum_{i,l}p(i,l|x,y)E_{i,l}\\
&=\Phi_p(E_{x,y}).
\end{align*}

Conversely,  assume that  $\Phi: M_n \to M_n$ is qc-factorizable via a quantum permutation $u=( e_{x,a}) \in M_n \otimes N$ where $N$ is a von Neuamnn algebra with a trace $\tau$.
Factorizability implies
 \begin{equation}\label{Equ-fact}
 \Phi(X)=id\otimes \tau(u^*(X\otimes 1_N)u)  \forall X\in M_n.
 \end{equation}
 
Evaluating $E_{x,y}$ in the Equation \ref{Equ-fact} in one hand and then comparing with the expression $\Phi(E_{x,y}) = \sum_{a,b} p(a,b|x,y) E_{a,b}$ on the other hand, we see that
\[ p(a,b|x,y) = \tau(e_{x,a}e_{y,b}).\]
Now using the fact that $u=(e_{x,a})$ is a quantum permutation, we have that $p \in C_{qc}^{bs}(n,n).$
\end{proof}

We now characterize the various t-factorizable maps. Note that a unital trace preserving completely positive map  of the form
\[X\mapsto \sum_{j=1}^k \lambda_j U_j^*XU_j,\]
is called a mixed unitary map where all the $U_j$'s are unitaries, $\lambda_j\geq 0$ for all $j$ and $\sum_j\lambda_j=1$.

If the unitaries of a mixed unitary map can be chosen to be permutations, then we call such a map a \textbf{mixed permutation map}.
Note that the set of mixed unitary maps is precisely the class of factorizable maps that admit an abelian ancilla, cf. \cite{hm11}. The following theorem exhibits similar characterization of the ancilla for the maps associated to the classical bisynchronous densities.
\begin{thm}
Let $p(a,b|x,y)$ be a bisynchronous density on $n$ inputs and $n$ outputs.  Then 
 $p\in C_{loc}^{bs}(n,n)$ if and only if  $\Phi_p:M_n \to M_n$ is a mixed permutation map. 
\end{thm}
\begin{proof}
 Let $\Phi_p(E_{x,y}):=\sum_{a,b}p(a,b|x,y)E_{a,b}=\sum_{j=1}^k \lambda_jU_j^*E_{x,y}U_j$,  where $U_j$'s are some permutations on $[n]$ and $\sum_j \lambda_j=1$. Note that for any permutation $\sigma$ on $[n]$ and the corresponding permutation  matrix $U=\sum_{z}E_{z,\sigma(z)}$, we have 
\[U^*E_{x,y}U=\sum_{z,w}E_{\sigma(z),z}E_{x,y}E_{w,\sigma(w)}=E_{\sigma(x),\sigma(y)}.\] 
This means that 
\[p(a,b|x,y)=\begin{cases}1 \ if \ \exists \ a \ permutation \ \ \sigma:  \sigma(x)=a,\sigma(y)=b, \\ 
0 \ otherwise.
  \end{cases}\]
  This is a deterministic correlation (see section 9.4 in \cite{vp}). Since the local-correlations are convex combinations of deterministic correlations, it follows that $p\in C_{loc}^{bs}(n,n)$. 

Conversely, let $p\in C_{loc}^{bs}(n,n)$. This means that Alice and Bob share some probability space $(\Omega,\mu)$ and a function $F:[n]\times \Omega\rightarrow [n]$ such that 
\[p(a,b|x,y)=\mu (\{w: F(x,w)=a, F(y,w)=b\}).\]
Now since $p$ is a bisynchronous density, it follows that almost surely $F$ is a bijective function on $[n]$. Indeed, if $x\neq y$ and $a=b$, then
\[0=p(a,a|x,y)=\mu (\{w: F(x,w)=a, F(y,w)=a\}).\]
This shows that the set where $F$ fails to be injective has measure zero. Now taking union of all the input and output sets which is a finite set, we get finite union of measure zero sets which is again a measure zero set. 
Let's call this set $\mathcal{Z}$.
Now for $w\in \Omega\setminus \mathcal{Z}$, we set $U(w)=\sum_{x}E_{x,F(x,w)}$. This is a permutation matrix. Now notice that 
\begin{align*}\Phi_p(E_{x,y})&=\sum_{a,b}p(a,b|x,y)E_{a,b}\\
&= \sum_{a,b} \int_{\{w:F(x,w)=a, F(y,w)=b\}}E_{F(x,w),F(y,w)}dw\\
&=\sum_{a,b}\int_{\{w:F(x,w)=a, F(y,w)=b\}}U(w)^*E_{x,y}U(w)dw \\
&= \int_{\Omega} U(w)^* E_{x,y} U(w) dw.
\end{align*}   
Now since there are only finitely many permutation matrices, the above integral is actually a convex sum. Since the matrix units generate the matrix algebra, we get $\Phi_p$ is a mixed permutation map.

 
\end{proof}

\subsection{The graph isomorphism game}
 Here we study the map $\Phi_p$ arising from a perfect strategy $p$ for the graph isomorphism game, introduced in \cite{mancinscaetal}. This map was also studied in  \cite{mancinscaetal2} and using this completely positive map, the authors obtained some very interesting results namely they have shown that various kinds of isomorphisms of two graphs is equivalent to having an isomorphism between the {\it coherent algebras associated to these graphs}. Here, we consider the factorizability properties of these maps. 
 
  In the following theorem we exhibit an interesting connection between perfect strategies of the graph isomorphism game and factorizable maps factoring via a quantum permutation. For two vertices $x,y$ in a graph $G$, we will use the symbol $x\sim y$ to denote that $(x,y)\in E(G)$, that is $(x,y)$ is an edge. Likewise, $x\nsim y$ will mean $(x,y)$ is not an edge.
  
\begin{thm}\label{thm-graph-iso-factorizable}
Let $G$ and $H$ be two graphs with $n$ vertices and adjacency matrices $A_G,A_H$, respectively, and let $t\in \{loc,q,qa,qc \}$. Then the following statements are equivalent:
\begin{enumerate}
\item $G\cong_{t} H$,
\item there exists a quantum t-permutation $(e_{g,h}):g \in V(G), h \in V(H)$ such that
\[(A_G\otimes 1)(e_{g,h})=( e_{g,h})(A_H\otimes 1).\] 

\item  there is a map $\Phi:M_n \to M_n$  that is t-factorizable map via a quantum permutation $u=(e_{x,a})$
with the property that 
$\Phi(A_G)=A_H$ and $\Phi^*(A_H)=A_G$, where $\Phi^*$ is the adjoint of $\Phi$. 
 
\end{enumerate}
\end{thm}
\begin{proof}
$1\Leftrightarrow 2$ follows from \cite[Theorem 4.9]{brannanetal} and \cite[Remark 2.5]{brannanetal}.

(\textbf{${2\Rightarrow 3}$})
We only do the case that t=qc.
Without loss of generality, embedding $\mathcal{A}$ into its double dual if necessary, we get a von Neumann algebra $N$ with a trace $\tau_N$, where the projections $\{e_{x,a}\}$ live in such that the magic permutation $u=(e_{x,a})\in M_n\otimes N$ intertwines the adjacency matrices. Now defining the map $\Phi(x)=id\otimes \tau_N (u^*(x\otimes 1_N)u)$, it is easy to see that
\[\Phi(A_G)=id\otimes \tau_N(u^*(A_G\otimes 1_N)u)=id\otimes \tau_N(A_H\otimes 1_N)=A_H.\] 
Also, following Theorem \ref{prop:factorizable} we get $\Phi(E_{x,y})=\sum_{a,b} \tau_N(e_{x,a}e_{y,b})E_{a,b}$. By symmetry, it follows that \[\Phi^*(E_{a,b})=\sum_{x,y}\tau_N(e_{a,x}e_{b,y})E_{x,y}.\]
And hence it is easy to see that $\Phi^*(A_H)=A_G$. Here we define $e_{a,x}=e_{x,a}$ and the corresponding densities arise from the flip of a original game introduced in subsection \ref{subsection-flip}.

($3\Rightarrow 1$) Here we will prove that all the winning conditions for the graph isomorphism game can be recovered from the data given in 3. We think of the input indices $x,y$ are taken from the graph $G$ and the outcome indices $a,b$ are from the graph $H$. 

Since we have $\Phi(A_G)=A_H$, we get 
\begin{align*}
\Phi(A_G)=\sum_{x\sim y}\Phi(E_{xy})&=\sum_{x\sim y}\sum_{a,b}\tau_N(e_{x,a}e_{y,b})E_{ab}\\
&=\sum_{x\sim y}\sum_{a\sim b}\tau_N(e_{x,a}e_{y,b})E_{ab}+\sum_{x\sim y}\sum_{a\nsim b}\tau_N(e_{x,a}e_{y,b})E_{ab}\\
&=A_H.
\end{align*}
Since the coefficients $\tau_N(e_{x,a}e_{y,b})$ are all positive numbers for any $x,a$, we see that if $x\sim y$ and $a\nsim b$, we must have $p(a,b|x,y)=\tau_N(e_{x,a}e_{y,b})=0$. So the connected edges must go to the connected edges. 

Now using the relation $\Phi^*(A_H)=A_G$ of the adjoint map, a similar argument shows that if $x\nsim y$ and $a\sim b$, we must have $\tau_N(e_{a,x}e_{b,y})=0$. Now using the flip of the graph isomorphism game we have $e_{x,a}=e_{a,x}$ and hence we get $p(a,b|x,y)=\tau_N(e_{x,a}e_{y,b})=0$. So the disconnected edges must go to the disconnected edges.

Note that if $x=y$, then 
\begin{align*}
\Phi(1)=\sum_{x}\Phi(E_{xx})&=\sum_{x}\sum_{a,b}\tau_N(e_{x,a}e_{x,b})E_{ab}\\
&=\sum_{x}\sum_{a}\tau_N(e_{x,a}e_{x,a})E_{aa}+\sum_{x\sim y}\sum_{a\neq b}\tau_N(e_{x,a}e_{x,b})E_{ab}\\
&=1.
\end{align*}
This shows that if $a\neq b$, then $p(a,b|x,x)=\tau_N(e_{x,a}e_{x,b})=0$. This shows that if the two inputs are same, then the two outputs can not be different.

Hence all the ``relations" in the graph isomorphism game are preserved if we have the assertion 3. Hence 1 follows.  
\end{proof}

\begin{remark}
Note that in \cite{hm11} and \cite{hm15}, in all the examples given for factorizable maps, the associated von Neumann algebras could be taken to be finite dimensional. In a very recent article, Musat and R\o rdam (\cite{MR18}), established examples of factorizable maps that do not admit factorization through any finite dimensional von Neumanna algebra and do require ancillas of type $\rm{II_1}$. 

 Note that in Corollary 7.4 of the article \cite{mancinscaetal}, the authors show that there exists graphs $G$ and $H$ on $n$ vertices such that $G\cong _{qc}H$ but $G\ncong_{q}H$. This was refined somewhat in \cite{KPS} where it was shown that Slofstra's example of a linear BCS game with a perfect qa-strategy but no perfect q-strategy \cite{Sl}, yields a pair of graphs such that $G \cong_{qa} H$ but $G \ncong_q H$.
 
 This latter fact implies by Theorem~\ref{thm-bs-charact.} that there are projections $\{e_{x,a}\}_{x,a}$ in an ultrapower of the hyperfinite $II_1$-factor $\cl R^{\omega}$ with trace $\tau_{\omega}$, which form a quantum permutation $u=(e_{x,a})$ such that  $(A_G\otimes 1)u=u(A_H\otimes 1)$.  Moreover, no quantum permutation that can be formed with projections coming from a finite dimensional C$^*$-algebra can satisfy this intertwining property.  
 
 If we consider the corresponding factorizable map on $M_n$ given by
 \[\Phi(X)= id\otimes \tau_{\omega}(u^*(X\otimes 1_{\cl R^{\omega}})u),\]
 then by Theorem \ref{thm-graph-iso-factorizable}, it is clear that $\Phi$ can not admit a factorization through any finite dimensional von Neumann algebra as long as the unitary $u=(u_{i,j})$ implementing the factorization is required to be a quantum permutation.   However, we have not been able to show that this map cannot factorize through a finite dimensional von Neumann algebra. 
\end{remark}
 \subsection{The Orbital algebra and fixed point algebra}
 Lupini et al. in \cite{lupini} start with a subgroup $\bb G$ of the quantum permutation group $\cl O(S_n^+)$ and the quantum permutation $u=(u_{ij})_{i,j=1}^n$ that generates the subgroup and define a subalgebra of $M_n$ that they call the {\it orbital subalgebra}. They 
 prove that this algebra is equal to the set of matrices $A \in M_n$ such that
 \[ (A \otimes 1) u = u(A \otimes 1).\]
 Moreover, they show that this algebra is also closed under the Schur product of matrices.
 We note that not every quantum permutation generates a subgroup. As seen in the work of Lupini et al, the quantum permutation needs to satisfy some additional properties.
 
 
 
 
 In this subsection, we consider the case of an arbitrary bisynchronous density $p$ for $n$ inputs and $n$ outputs. By our characterization theorem, each such density comes from a (possibly non-unique) quantum permutation.
 We give a new characterization of the above "commutant", that applies to more general quantum permutations.
 
 Note that to perform the multiplication $(A \otimes 1)(u_{x,a}) = (u_{x,a}) (A \otimes 1)$ we need to not just have that the input and output sets have the same size, but to have them identified as the same set.  This just corresponds to fixing an identification $I= O = \{ 1,..., n\}$, so that we may write our bisynchronous correlation as $p(a,b | x,y), 1 \le a,b,x,y \le n$.
 
  The new condition is related to the fixed points of the bisynchronous map $\Phi_p$. Note that since $\Phi_p$ is a unital and trace preserving completely positive map, the fixed point set 
 \[Fix(\Phi_p):=\{A\in M_n: \Phi_p(A)=A\},\]
 is a subalgebra of $M_n$ and moreover, if $\Phi_p$ is represented with $m$ linearly independent Kraus operators $\{K_i\}$:
 \[\Phi_p(X)=\sum_{i=1}^m K_i^*XK_i, \ K_i\in M_n \ \forall i,\]
 then $A\in Fix(\Phi_p)$ if and only if $AK_i=K_iA, \  and \  AK_i^*=K_i^*A \ \forall i$ by a result of Kribs \cite{kribs}.
  
 \begin{thm}
 Let $p\in C_{qc}^{bs}(n,n)$, let $\Phi_p: M_n \to M_n$ be the associated bisynchronous map, let $u=(e_{x,a})_{a,x=1}^n$ be a quantum permutation and let $\tau$ be a faithful trace on the algebra $\cl A$ generated by the entries of $u$ such that $p(a,b|x,y) = \tau(e_{x,a} e_{y,b})$. Then the following statements about a matrix $A\in M_n$ are equivalent:
 \begin{enumerate}
 \item $(A\otimes 1_{\cl A})u=u(A \otimes 1_{\cl A}),$
 \item $A\in Fix(\Phi_p)$,
 \item $A=(a_{i,j})$ with $a_{i,j}= a_{k,l}$ whenever $e_{i,k}e_{j,l} \ne 0$.
 \end{enumerate}
 \end{thm}
 \begin{proof}
 Note that $(1)\implies (2)$ follows from the fact that $\Phi_p$ is factorizable. Indeed, from Theorem \ref{prop:factorizable} it follows that 
 \[\Phi_p(X)=id\otimes \tau(u^*(X\otimes 1)u),\forall X\in M_n.\]
  Now if $A\in M_n$
satisfies the condition in $(1)$, then it is easily seen that $\Phi_p(A)=A$.

To see that (2) implies (1), let $\Phi_p(A)=A$ and suppose the Choi representation of $\Phi_p$ is given by
\[\Phi_p(X)=\sum_{i=1}^m K_i^*XK_i, \ K_i\in M_n \ \forall i,\]
where $\{K_i\}_{i=1}^m$'s are linearly independent.
Then since $\Phi_p$ is a factorizable map, following Theorem 2.2 in \cite{hm11}, there exist operators $v_1,\cdots,v_m\in \cl A$ such that $u=\sum_{i=1}^m K_i\otimes v_i$ and $\tau(v_i^*v_j)=\delta_{ij}, 1\leq i,j\leq m.$

Now since $\Phi_p(A)=A$ if and only if $AK_i=K_iA$, for all $i=1,\cdots,m$, we have 
\[(A\otimes 1)u=(A\otimes 1)(\sum_{i=1}^m K_i\otimes v_i)=(\sum_{i=1}^m K_i\otimes v_i)(A\otimes 1)=u(A\otimes 1).\]

Finally, we prove that $(A \otimes 1) (e_{x,a}) = (e_{x,a})(A \otimes 1)$ if and only if $A= (a_{i,j})$ satisfies the condition of (3).
First assume that $e_{i,k}e_{j,l} \ne 0$. Considering the $(i,l)$-entry of the product we have that
\[ e_{i,k} ((A \otimes 1)(e_{x,y}))_{i,l} e_{j,l} = e_{i,k} ( \sum_r a_{i,r} e_{r,l}) e_{j,l} = e_{i,k} a_{i,j} e_{j,l},\]  since $e_{r,l} e_{j,l} =0$ unless, $r=j$.
On the other hand,
\[ e_{i,k} ((e_{x,y})(1 \otimes A))_{i,l} e_{j,l} = e_{i,k} (\sum_r e_{i,r} a_{r,l}) e_{j,l} = e_{i,k} a_{k,l} e_{j,l}.\]

Thus, if the commutation holds, then since the product $e_{i,k} e_{j,l} \ne 0$, we must have that
$a_{i,j} = a_{k,l}$.

Conversely, assume that the set of equalities holds. 
The $(i,k)$-th entry of  $u^*(A \otimes 1) u$ is
\[  \sum_{r,l} e_{l,i} a_{l,r} e_{r,k} = a_{i,k}( \sum_{l,r} e_{l,i} e_{r,k}) =a_{i,k} 1,\]
where we have used that any time that $e_{l,i}e_{r,k} \ne 0$ then $a_{l,r} = a_{i,k}$ and the fact that
$\sum_l e_{l,i} = \sum_r e_{r,k} =1$. From this calculation it follows that $u^*(A \otimes 1) u = A \otimes 1$, and hence, $(A \otimes 1) u = u(A \otimes 1)$, since $u$ is a unitary. 
\end{proof}
\begin{cor} Let $p\in C^{bs}_{qc}(n,n)$ then the fixed point algebra $Fix(\Phi_p)$  is closed under Schur product.
\end{cor}
\begin{proof} If the matrices $A=(a_{i,j})$ and $B= (b_{i,j})$ satisfy the conditions of (3), then so does the matrix $(a_{i,j}b_{i,j}).$
\end{proof}

The above proof of the equivalence of (2) and (3) borrows heavily from the proofs in Lupini et al (see \cite{lupini} ), but differs in several distinct ways. Their proof starts by defining equivalence relations on the set $[n]$ and $[n] \times [n]$ induced by the matrix $u$. In our more general setting, these are no longer equivalence relations.

The above result also shows a small amount of uniqueness among the quantum permutations that can be used to represent a bisynchronous correlation.

 \begin{cor} Let $p \in C^b_{qc}(n,n)$, let $\cl A, \cl B$ be C*-algebras with faithful traces $\tau$ and $\rho$, respectively, and let $u_{i,j} \in \cl A$ and $v_{i,j} \in \cl B$ form quantum permutations such that $p(a,b|x,y) = \tau(u_{x,a}u_{y,b}) = \rho(v_{x,a}v_{y,b})$.
 Then for $A \in M_n$ we have that
 \[ (A \otimes 1_{\cl A}) (u_{x,a}) = (u_{x,a})(A \otimes 1_{\cl A}) \iff (A \otimes 1_{\cl B}) (v_{x,a}) = (v_{x,a})(A \otimes 1_{\cl B}).\]
 \end{cor}

\section{Problems and  future direction}

Ozawa \cite{Oz13} proved that Connes' Embedding Problem has a positive answer if and only if $C_{qa}(n,k) := C_q(n,k)^- = C_{qc}(n,k), \forall n,k$, where the bar denotes closure.  Later \cite{DP} proved the analogous result for synchronous correlations. Namely,  that Connes' Embedding Problem has a positive answer if and only if $C^s_{q}(n,k)^- = C^s_{qc}(n,k), \forall n,k.$  At the time the relationship between $C^s_{qa}(n,k)$ and the closure of $C^s_q(n,k)^-$ was unknown.  This {\it synchronous approximation problem,} i.e., whether or not every synchronous correlation that is a limit of correlations in $C_q(n,k)$ is actually a limit of synchronous correlations in $C_q^s(n,k)$ was settled in \cite{KPS} where it was shown that, indeed, $C^s_q(n,k)^- = C_{qa}^s(n,k)$.

We do not know the answer to either of these questions for bisynchronous correlations. 

\begin{prob} Is the answer to Connes' Embedding Problem positive if and only if $C^{bs}_q(n,n)^- = C^{bs}_{qc}(n,k), \forall n$? Or possibly if and only if $C^{bs}_{qa}(n,n) = C^{bs}_{qc}(n,n), \forall n$ ?
\end{prob}
If the answer to this first problem is affirmative, then it would show a close connection between the quantum permutation group and Connes' Embedding Problem. In a recent article (see \cite{CEP}) it was announced that the Connes' Embedding Problem has a negative answer. It will be interesting to see if there is any consequence of this result related to quantum permutation groups and bisynchronous games.

\begin{prob} Is $C^{bs}_{q}(n,n)^- = C^{bs}_{qa}(n,n), \forall n$ ?
\end{prob}
This latter problem is closely related to the following approximation problem for ``almost" quantum permutations.
Given $A \in M_d$ we set $\|A\|^2_2= \frac{1}{d} Tr(A^*A)$.
\begin{prob} Given $\epsilon >0$ is there a $\delta= \delta(\epsilon) > 0$, independent of $d$, such that whenever $E_{i,j} = E_{i,j}^2 = E_{i,j}^* \in M_d, \, 1 \le i,j \le n$ satisfy 
\[ \| I_d - \sum_{j=1}^n E_{i,j} \|_2 < \epsilon \, \forall i \text{ and } \| I_d - \sum_{i=1}^n E_{i,j} \|_2 < \epsilon \, \forall j,\]
then there exists a quantum permutation $u=(F_{i,j})$ with $\|E_{i,j} - F_{i,j}\|_2 \le \delta \, \forall i,j$ and such that $\delta \to 0$ as $\epsilon \to 0$  ?
\end{prob}

In \cite{brannanetal} it was shown that the *-algebra of the graph isomorphism game always admits a representation onto a Hilbert space and moreover it admits a representation onto a C$^*$-algebra with a trace. Hence, whenever this *-algebra is non-zero, then the corresponding graph isomorphism game has a perfect qc-strategy.  It is natural to ask if the same holds for bisynchronous games, in general.

\begin{prob} Let $\cl G$ be a bisynchronous game. If $\cl A(\cl G) \ne (0)$, then does this algebra always admit a non-trivial *-representation on a Hilbert space? It is also natural to ask whether this algebra admits a trace, i.e., does it possess a unital *-homomorphism into a C*-algebra with a trace ?
\end{prob}
 


\section{Acknowledgements}
This work was done while MR was a Postdoctoral Fellow at the department of Pure Mathematics, University of Waterloo.

\end{document}